\documentclass{article}
\usepackage[utf8]{inputenc}
\usepackage{amsmath}
\usepackage{amsthm}
\usepackage{amssymb}
\usepackage{pgf}
\usepackage{tikz}
\usetikzlibrary{automata}
\usetikzlibrary{decorations,arrows}
\usetikzlibrary{decorations.pathmorphing}
\usetikzlibrary{positioning}
\usepackage[british]{babel}

\newcommand\definesymb[1]{%
\expandafter\newcommand\csname #1#1\endcsname{{\ensuremath{\mathbb{#1}}}}%
}

\definesymb{Z}
\definesymb{N}
\definesymb{R}

\newtheorem{theorem}{Theorem}
\newtheorem{defn}{Definition}[section]
\newtheorem{proposition}[defn]{Proposition}
\newtheorem{lemma}[defn]{Lemma}
\newtheorem{fact}[defn]{Fact}
\newtheorem{cor}[defn]{Corollary}

\newtheorem{conj}{Conjecture}
\title{Translation-like Actions and Aperiodic Subshifts on Groups}

\author{Emmanuel Jeandel\\
LORIA, UMR 7503 - Campus Scientifique, BP 239\\
54\,506 VANDOEUVRE-L\`ES-NANCY, FRANCE\\
\texttt{emmanuel.jeandel@loria.fr}}

\begin{document}

\maketitle

\begin{abstract}

It is well known that if $G$ admits a f.g. subgroup $H$ with a weakly
aperiodic SFT (resp.  an undecidable domino problem), then $G$
itself has a weakly aperiodic SFT (resp. an undecidable domino problem).
We prove that we can replace the property ``$H$ is a subgroup of $G$''
by  ``$H$ acts translation-like on $G$'', provided $H$ is finitely presented.

In particular:
\begin{itemize}
	\item If $G_1$ and $G_2$ are f.g. infinite groups, then $G_1
	  \times G_2$ has a weakly aperiodic SFT (and actually a
	  undecidable domino problem). In particular the Grigorchuk group
	  has an undecidable domino problem.	  
	\item Every infinite f.g. $p$-group admits a weakly
	  aperiodic SFT.
\end{itemize}
\end{abstract}

A subshift of finite type over a group $G$ corresponds to a
description of  colorings of the vertices of its Cayley graph subject
to local constraints.
Even for harmless groups like $\mathbb{Z}^2$, it is possible to build
\cite{Berger2} easily \cite{Kari14} subshifts with no periodic points.
For this group, the domino problem, which consists in deciding if a
subshift of finite type is empty, is even undecidable \cite{Berger2}.

In this article, we are interested in which groups enjoy similar
properties: In which groups can we build aperiodic subshifts of finite
type, and which groups have an undecidable word problem.
There are various definitions of ``aperiodicity'' on a group, and here
we study weakly aperiodic subshifts: no coloring has a finite orbit.

Apart from the example of $\mathbb{Z}^2$, aperiodic subshifts have
been built on  Baumslag-Solitar Groups \cite{AuKa}, on the free
group \cite{Piantadosi} and on every group of nonlinear polynomial growth
\cite{BallierStein,Carroll}.
In all these examples except the free group, this actually gives
groups with an undecidable domino problem.

It is easy to see that if a f.g. group $G$ has an aperiodic SFT, then
every group that contains $G$ also has an aperiodic SFT.
In this statement, ``contains'' means subgroup containment. The goal
of this article is to prove that this is also true in a stronger
sense. We say that $H$ acts translation-like on $G$ if, upto the
deletion of some edges, some Cayley graph of $G$ can be partitioned
into copies of the Cayley graph of $H$.

We will then prove that if $H$ is finitely presented, acts
translation-like on $G$, and has an aperiodic SFT, then $G$ also has
an aperiodic SFT.
As a corollary, we are able to show that the direct product of two
infinite f.g. groups has an aperiodic SFT, and that any nonamenable
f.g. group has an aperiodic SFT.

The main idea from this article comes from the work of Ballier and
Stein \cite{BallierStein}, which explains (in particular) how to build aperiodic SFTs
on $\mathbb{Z} \times G$ as soon as $G$ is f.g infinite, by producing
copies of $\mathbb{Z}$ inside $G$. We generalize this statement by
searching for copies of other groups inside $G$.
The main idea is that if $G$ contains copies (in the sense of
translation-like actions) of some group $H$, and if $H$ is finitely
presented, then we can find and describe these copies by a finite
number of local constraints.

The idea is also reminiscent of Cohen \cite{Cohen2014} who proved that
having a weakly aperiodic SFT is a quasi-isometry invariant for
finitely presented groups.
Two groups $G$ and $H$ are quasi-isometric if their Cayley graph look
the same from a distance. This notion is somehow related to
translation-like action, but there are differences: a quasi-isometry
cannot distort the distances too much, and a translation-like action
cannot identify two vertices of the Cayley Graph, so that these
notions are ultimately different (The article \cite{Dymarz} possibly
provides an example of quasi-isometric groups $G$ and $H$ s.t. $G$
does not act translation-like on $H$).
The proof method used by Cohen is similar to ours: Cohen encode in a
subshift in $G$ local information that is sufficient to exhibit a
quasiisometry from $G$ to $H$. We do the same with translation-like
actions, which is ultimately much simpler.

\section{Definitions}
We assume some familiarity with group theory and actions of groups.
See \cite{CicCoo} for a good reference on symbolic dynamics on groups.

All groups below are implicitely supposed to be finitely generated
(f.g. for short).

The notion of a Cayley graph is used throughout this article to give
some intuitions, but is not technically needed.
If $G$ is a group with generators $S$, the Cayley graph $C(G;S)$ is
the graph with vertices $G$, and edges $(g,gs)$ for $s \in S$.

The Baumslag-Solitar groups $B(m,n)$ is the group $B(m,n) =
\left<a,b|ab^ma^{-1} = b^n\right>$. When we speak of a
	Baumslag-Solitar group, we implicitely suppose that both $m$ and $n$
	are nonzero (they might be negative).
We will be interested mainly in the groups $B(1,n)$. $B(1,1)$ is in
particular the group $\mathbb{Z}^2$.

\clearpage
\subsection{Symbolic dynamics on groups}

Let $A$ be a finite set and $G$ a group.
We denote by $A^G$ the set of all functions from $G$ to $A$.
For $x \in A^G$, we write $x_g$ instead of $x(g)$ 
for the value of $x$ in $g$.

$G$ acts on $A^G$ by
\[
(g \cdot x)_h = x_{g^{-1}h}
\]

A \emph{pattern} is a partial function $P$ of $G$ to $A$ with finite
support. The support of $P$ will be denoted by $\mathrm{Supp}(P)$.

A \emph{subshift} of $A^G$ is a subset $X$ of $A^G$ which is
topologically closed (for the product topology on $A^G$) and invariant
under the action of $G$.

A subshift can also be defined in terms of forbidden patterns.
If $\cal P$ is a collection of patterns, the subshift defined by $P$ is 

\[
X_{\cal P} = \left\{ x\in A^G | \forall g \in G, \forall P \in {\cal P}
  \exists h \in \mathrm{Supp}(P), (g\cdot x)_h \not= P_h \right\}
\]
Every such set is a subshift, and every subshift can be defined this way.
If $X$ can be defined by a finite set $\cal P$, $X$ is said to be a
subshift of finite type, or for short a SFT.

For a point $x \in X$, the stabilizer of $x$ is $Stab(x) = \{ g | g
  \cdot x =   x\}$
A point $x$ is strongly periodic if $Stab(x)$ is a subgroup of $G$ of
finite index. Equivalently, the orbit $G \cdot x$ of $x$ is finite.

A subshift $X$ is \emph{weakly aperiodic} if it is nonempty and does not contain any
strongly periodic point.
In the remaining, we are interested in groups $G$ which admit 
weakly aperiodic SFTs

A f.g. group $G$ is said to have decidable domino problem if there is an
algorithm that, given a description of a finite set of patterns $\cal
P$, decides if $X_{\cal P}$ is empty.
It is easy to see that f.g. groups with undecidable word problem have
trivially an undecidable domino problem, so this question is mostly
relevant for groups with a decidable word problem.

These two properties are related in the following way: If every
(nonempty) SFT over $G$ has a strongly periodic point (and $G$ has
decidable word problem), then $G$ has decidable domino problem. 

Finding which groups have a weakly aperiodic SFT, and which groups have
decidable domino problems is the topic of this article.

We now summarize previous theorems:
\begin{theorem}
\begin{itemize}
	\item $\mathbb{Z}$ does not admit a weakly aperiodic SFT and has decidable word problem
	\item Free groups of rank $\geq 2$ admit weakly aperiodic SFTs
	  \cite{Piantadosi} and have decidable word problem \cite{Kuske}
	\item The free abelian group $\mathbb{Z}^2$ \cite{BergerPhD}, the Baumslag Solitar groups
	  \cite{AuKa} admit weakly aperiodic SFTs, and have an undecidable word problem
	\item f.g. Nilpotent groups have weakly aperiodic SFTs and an
	  undecidable word problem unless they are virtually cyclic  \cite{BallierStein, Carroll}
  \end{itemize}		
\end{theorem}	
We note in passing that these questions may be asked more generally
for tilings of translation surfaces rather than coloring of groups, for example of
$\mathbb{R}^2$, the hyperbolic plane \cite{Kari99}, amenable spaces
\cite{Block}, or manifolds of intermediate growth \cite{marcin}.
How these results translate into aperiodic SFTs on groups is not
clear. We will provide group versions of the two last theorems in this
paper.

There are a few general statements on these questions which work as
follow:
\begin{theorem}
\begin{itemize}	
	\item Let $G,H$ be f. g. commensurable groups. Then $G$ admits a
	  weakly aperiodic SFT (resp. has an undecidable domino problem)
	  if only if $H$ does. \cite{Carroll}
	\item Let $H$ be a f.g. normal subgroup of $G$ f.g. If $G/H$
	  admits a weakly aperiodic SFT (resp. has an undecidable domino
	  problem), then $G$ does.
	\item Let $G,H$ be \emph{finitely presented} groups that are
	  quasi-isometric. Then $G$ admits a
	  weakly aperiodic SFT (resp. has an undecidable domino problem)
	  if and only if $H$ does. \cite{Cohen2014}
	\item Let $G,H$ be two finitely presented groups that are
	  quasi-isometric but not commensurable. Then $G$ and $H$ admit
	  weakly aperiodic SFT \cite{Cohen2014}
\end{itemize}	  
\end{theorem}	
The relevant articles above \cite{Carroll,Cohen2014} only deal with
weakly aperiodic SFTs (or strongly aperiodic SFTs) but the results
about undecidable domino problems may be obtained in the same manner.

We conclude this review by the following easy result:
\begin{proposition}
	Let $H \subseteq G$ be f.g. groups. 
	If $H$ has a weakly aperiodic SFT (resp. an undecidable domino
	problem), then $G$ does.
\end{proposition}
\begin{proof}
	Let $X_H$ be a SFT on $H$.
\[
X_H = \left\{ x\in A^H | \forall h \in H, \forall P \in {\cal P}   \exists p \in \mathrm{Supp}(P), (h\cdot x)_p \not= P_p \right\} \]

	Taking the same forbidden patterns, we obtain a subshift $X_G$ on $G$.
	
\[
X_G = \left\{ x\in A^G | \forall g \in G, \forall P \in {\cal P}   \exists p \in \mathrm{Supp}(P), (g\cdot x)_p \not= P_p \right\} \]

If $X_G$ is nonempty, then $X_H$ is nonempty: The restriction of $x \in X_G$ to $A^H$ is in $X_H$.
If $X_H$ is nonempty, then $X_G$ is nonempty: Take $x \in X_G$, writee $G = HK$ for some
transversal $K$ and define $y_{hk} = x_h$.

This proves that $G$ has an undecidable domino problem if $H$ does.

Now if $X_G$ has a point $x$ with stabilizer $N$ of finite index in $G$, then the
restriction of $x$ to $A^H$ is a point of $X_H$ with stabilizer that
contains $H \cap N$, hence is of finite index in $H$.
Hence $G$ has a weakly aperiodic SFT if $H$ does.
\end{proof}
From all previous theorems, virtually cyclic groups are the only
groups for which we are able to prove that they don't have aperiodic
SFT, and virtually free groups the only groups for which we are able
to prove that they have decidable domino problem.
The conjectures below state these are the only cases:

\begin{conj}[\cite{BallierStein}]
A f.g. group $G$ has a decidable domino problem iff it is virtually free.   
\end{conj}

\begin{conj}[\cite{Carroll}]
A f.g. group $G$ has no weakly aperiodic SFT iff it is virtually cyclic.
\end{conj}

\subsection{Translation-like actions}

The concept of translation-like action was introduced by Whyte \cite{Whyte}.
We will use here an alternative definition which is a compromise
between the original definition from Whyte and Cor~5.2 in Seward \cite{Seward}.

\begin{defn}
Let $H$ and $G$ be f.g groups. Let $S_H$ be a generating set for $H$.

We say that $H$ acts translation-like on $G$ if $H$ right acts on $G$ s.t:
\begin{itemize}
  \item The action is free: $g \circ h = g$ for some $g$ implies $h =\lambda_H$.
  \item There exists a finite set $S_G$ and a map $\phi: G \times S_H
\rightarrow S_G$ s.t. ${g \circ h = g \phi(g,h)}$.
\end{itemize}
Said otherwise, there exists a finite set $S_G$ and a partial
labelling of the edges of the Cayley graph $C(G,S_G)$ by elements of $S_H$
s.t. the restrictions of $C(G,S_G)$ to the labelled edges is the
disjoint union of copies of the Cayley graph of $C(H,S_H)$.
\end{defn}
Note that the definition does not depend on the choice $S_H$ of a
generating set for $H$.

As hinted in \cite{Whyte,Seward}, translation-like actions generalize
(for f.g. groups) subgroup containment: If $H$ is a subgroup of $G$, 
then $H$ acts translation-like on $G$. More generally:
\begin{lemma}
	Let $G,H,N$ be f.g. groups.
	
	If $H$ acts translation-like on $N$ and $N$ acts translation-like
	on $G$,
	then $H$ acts translation-like on $G$.
\end{lemma}	
\begin{proof}	
This is obvious from the definition in terms of Cayley graphs.

For an actual proof, let {\scriptsize $\square$} be the action of $H$
on $N$ (witnessed by the sets $S_H$ and $S_N$) and
 $\circ$ be the action of $N$ on $G$ (witnessed by the sets $S_N$ and
 $S_G$, recall we can take the same set $S_N$)

Write $G = K \circ N$ for some transversal $K$ (which exists by
freeness of the action) and define the action by
$(k\circ i) \star h = k \circ (i${\scriptsize $ \square$}$ h)$.

It is indeed an action.
Furthermore, let $h \in S_H$ and $g \in G$. Write $g = k \circ i$.
$(k \circ i) \star h = k \circ (i${\scriptsize $ \square$}$ h) = k
\circ (in)$ for some $n \in S_N$ as {\scriptsize $\square$} is translation-like,
and $k \circ (in)= (k \circ i) g'$ for some $g' \in S_G$ as $\circ$ is
translation-like. Therefore $g \star h = gg'$ for some $g' \in S_G$.
\end{proof}

\clearpage
\section{The construction}
Before going into the proof, we start with a warmup.

Let $G$ be an f.g. infinite group and consider its Cayley graph $C(G;S)$
with respect to some finite set of generator $S$.
Now obtain a subgraph of $C(G;S)$ by keeping only one outgoing edge and
one incoming edge for each vertex.
This subgraph $\cal G$ is therefore an union of biinfinite paths and circuits.
Furthermore, it is always possible (if $S$ has been chosen carefully)
for $\cal G$ to consist only of biinfinite paths (This is nontrivial
and comes from the fact that $\mathbb{Z}$ acts translation-like on any
f.g. infinite group).
In other words, $\cal G$ is the union of copies of $\mathbb{Z}$.
The choice of an incoming and an outgoing edge can be simulated easily
by a subshift of finite type, so that we have proven that for any f.g. 
infinite group $G$, there is a SFT $S$ where every element of $S$
somehow partitions the Cayley graph of $G$ into copies of the Cayley
graph of $\mathbb{Z}$ or of  quotients of $\mathbb{Z}$, and some
element of $S$ partitions the Cayley graph of $G$ into copies of
the Cayley graph of $\mathbb{Z}$ only.

This construction is already sufficient to prove that $G_1 \times G_2$
admits an aperiodic SFT whenever $G_1$ and $G_2$ are f.g. infinite.
To do better, we could try to embed something other than $\mathbb{Z}$,
for example $\mathbb{Z}^2$, or the free group $\mathbb{F}_2$.

The general idea is as follows: Let $H$ be a f.g. group. We want to build
a SFT $X$ over a group $G$ s.t.
\begin{enumerate}
	\item Each element of $X$ somehow partitions the Cayley graph of
	  $G$ into copies of the Cayley graphs of $H$ or of quotients of $H$.
	\item There exist an element of $X$ for which the partition
	  correspond to copies of the Cayley graph of $H$ itself.
\end{enumerate}	

The good notion to make the second point work is the concept of
translation-like action.
With a subshift of finite type, we can only test properties of
the Cayley graph in a finite neighborhood of every point, this is the
reason why we will require $H$ to be finitely presented.

\subsection{Some Definitions} 

We start with some notations. Let $H$ be a finitely presented group
that we see as a finitely presented \emph{monoid}
 $H = \left< h_1 \dots h_n | R_1 \dots R_p\right>$.
Let $S = \{ h_1 \dots h_n\}$. Let $\cal R$ be the set of relations seen as a subset of $S^\star \times S^\star$.

In the rare circumstances when it will be necessary to differentiate an
element of $S^\star$ from the corresponding element of $H$, we will
write $h$ for an
element of $S^\star$ and $\underline{h}$ for the corresponding element of $H$.

Let $G$ be a f.g. group and $S_G$ be a finite subset of $G$ s.t. $H$
acts translation like on $G$ for this choice of  $S_G$.

\paragraph{}
We denote by $F$ the set of functions from $S$ to $S_G$.
Let $\Sigma$ be a finite alphabet.
In the remaining we are interested in points of $(\Sigma \times F)^G$.
For $x \in (\Sigma \times F)^G$, we write $\sigma(x_g)$ for the
$\Sigma$-component of $x_g$ and $x_g(h_i)$ for the value of the
function in the $F$-component of $x_g$ at  $h_i$.
In other words  $\sigma(x_g)$ is a notation for $(x_g)_1$ and $x_g(h_i)$ a
notation for $(x_g)_2(h_i)$.

A point $x$ should be interpreted on the Cayley graph $C(G,S_G)$ of $G$.
$\sigma(x_g)$ is the color at vertex $g$.
$x_g(h_i)$ is the edge of $C(G,S_G)$ we have to follow from vertex $g$ 
if we want to simulate going into direction $h_i$ in the Cayley graph of $H$.

\subsection{Preliminaries}

For now on, we concentrate on the edges, and we will deal with the
symbols later on.
Let $x \in (\Sigma \times F)^G$.
For $g$ in $G$ and $h \in S^\star$, we denote by $\Phi(g,x,h)$ the
element of $G$ obtained starting from $g$ and following the edges
(given by $x$) corresponding to $h$.

Formally $\Phi$ can be defined by :
\begin{itemize}
\item $\Phi(g,x,\epsilon) = g$ (where $\epsilon$ is the empty word of
  $S^\star$)
\item $\Phi(g,x, h_i h) = \Phi(g x_g(h_i) ,x, h)$
\end{itemize}
Note that, for each element $x$,  this defines an action of $S^\star$
into $G$, that is:
\begin{fact}
	\label{fact:action}
$\Phi(g,x,hh') = \Phi(\Phi(g,x,h),x, h')$.
\end{fact}	
Another way to see $\Phi$ is as follows (this can be proven by a
straightword induction on the length of $h$):
\begin{fact}
$	\Phi(g, g \cdot x, h) = g\Phi(\lambda_G, x, h)$.
\end{fact}

The main object that will interest us in this section is the following
subshift:
\[
	\begin{array}{rcl}
	X &=& \left\{ x \in (\Sigma \times F)^G | \forall g \in G, \forall
	  (h, h') \in {\cal R},
	  \Phi(\lambda_G, g \cdot x, h) = 	  \Phi(\lambda_G, g \cdot x, h')\right\} \\	&=&\left\{ x \in (\Sigma \times F)^G | \forall g \in G, \forall (h,h') \in {\cal R},
	  \Phi(g, x, h) = \Phi(g,x,h')\right\} 
	\end{array}
     \]
The two conditions are equivalent by the second fact.

Note that if we denote $Z = \left\{ x \in (\Sigma \times F)^G |
  \forall (h,h') \in {\cal R},
  \Phi(\lambda_G, x, h) =   \Phi(\lambda_G, x, h') \right\}$, then $X = \{ x |
  \forall g, g\cdot x \in Z\}$, which proves that $X$ is indeed a subshift.
As  ${\cal R}$ and $A$ are finite, whether $x \in Z$ depends only on
the value of $x$ on a finite neighborhood of $\lambda_G$, hence $X$ is
a SFT.

$X$ is always well-defined, whether $H$ acts translation-like on $G$
or not. However $X$ can be empty.

Now, by its definition, $X$ has the following obvious property.
\begin{fact}
Let $h, h' \in S^\star$ s.t. $\underline{h} = \underline{h'}$ (i.e. $h$ and
$h'$ represent the same element in $H$).

Then for $x \in X$, $\phi(g, x, h) = \phi(g, x, h')$.
\end{fact}

\paragraph{}

In the remainder of this section, we introduce the key lemmas, which
explain how  any element of $X$ can be seen as an element of
$\Sigma^H$, and any element of $\Sigma^H$ can be seen as an element of
$X$ with a specific property.

\begin{defn}
Let $F: X \rightarrow \Sigma^H$ defined by 
$F(x)_{\underline{h}} = \sigma(x_{\Phi(\lambda, x, h)})$
\end{defn}
That is $F(x)_{\underline{h}}$ is the symbol we obtain starting from $\lambda_G$ and
following the direction given by $h$. As $x \in X$, this is well defined.

\begin{lemma}
	\label{lemme:sens1}	
	Let $x \in X$. Then for every $h \in H$, there exists
	$g \in G$ s.t. $h \cdot F(x) = F(g \cdot x)$.
\end{lemma}	
\begin{proof}
	For $\underline{h} \in H$, take $g =\phi(\lambda, x, h)$.
\end{proof}	

\begin{lemma}
	\label{lemme:sens2}
	Let $y \in \Sigma^H$. Then there exists $x \in X$ s.t. $F(x) = y$ and
	for any $g \in G$ there exists $h \in H$ s.t.
	$h \cdot F(x) = F(g\cdot x)$
\end{lemma}	
\begin{proof}
By definition $H$ acts translation like on $G$ by $\circ$.
Let $T \subseteq G$ be a set of representatives of this free action,
that is for every $g \in G$ there are unique $t \in T, \underline{h} \in H$ 
s.t. $t \circ \underline{h} = g$.
We suppose that $\lambda_G \in T$.

We now define $x$ in the following way:
$\sigma(x_{t\circ \underline{h}}) = y_{\underline{h}}$ and
$x_{g}(h_i) = \phi(g,\underline{h_i})$.

By the definition of $x$, we have $\Phi(g, x, h_i) = g \circ
\underline{h_i}$, therefore we have
$\Phi(g, x,h) = g \circ \underline{h}$.
In particular $x \in X$.

Now  $F(x) = y$. Indeed ${F(x)_{\underline{j}} =
  \sigma(x_{\Phi(\lambda, x,j)}) = \sigma(x_{\lambda
	\circ \underline{j}}) = y_{\underline{j}}}$.

Now let $g \in G$ and write $g^{-1} = t \circ h^{-1}$ for
some $t,h$.
Then
\newcommand\place{\vrule height 0pt depth 12pt width 0pt}
 \[\begin{array}{r@{}c@{}l@{}c@{}l@{}c@{}l}
	 F(g \cdot x)_{\underline{j}} &=& \place \sigma( (g \cdot
	 x)_{\Phi(\lambda,g \cdot x, j)}) &=& \sigma(
x_{g^{-1}\Phi(\lambda, g \cdot x, j)}) &=& \sigma(x_{\Phi(g^{-1},x,j)}) =
\sigma(x_{g^{-1} \circ \underline{j}}) \\&=& \sigma(x_{t \circ (h^{-1}
  \underline{j})}) &=&
y_{h^{-1}\underline{j}} &=& (h \cdot y)_{\underline{j}}
\end{array}
\]
Therefore $F(g \cdot x) = h \cdot F(x)$.
	
\end{proof}

\subsection{The proof}

We now start from a SFT $X_H \subseteq \Sigma^H$ and build a SFT
$X_G \subseteq (\Sigma \times F)^G$ s.t.:

\begin{itemize}
	\item $X_G$ is empty iff $X_H$ is empty.
	\item If $X_H$ is weakly aperiodic, then $X_G$ is weakly aperiodic.
\end{itemize}	

Let $\cal P$ be a finite collections of patterns s.t.

\[ X_H = \{ y \in \Sigma^H | \forall h \in H, \forall P \in {\cal P}
	  \exists j \in \mathrm{Supp}(P), (h \cdot y)_j \neq P_j \}\]
	
Now we define 
\[ X_G = \{ x \in X | \forall g \in G, \forall P \in {\cal P}
	  \exists j \in \mathrm{Supp}(P), F(g \cdot x)_j \neq P_j \}
\]	
As before, $F(x)_j$ depends only on a finite neighborhood of the
identity, thus $X_G$ is a subshift of finite type.

\begin{lemma}
	If $X_G$ is nonempty, then $X_H$ is nonempty.
	More precisely, for any $x \in X_G$, $F(x) \in X_H$.
\end{lemma}	
\begin{proof}
	Obvious by lemma \ref{lemme:sens1}.
\end{proof}
\begin{lemma}
If $X_H$ is nonempty, then $X_G$ is nonempty
\end{lemma}	
\begin{proof}
	Obvious by lemma \ref{lemme:sens2}.
\end{proof}

\begin{lemma}
	\label{lemme:sens3}
	If $X_H$ is weakly aperiodic, then $X_G$ is weakly aperiodic.	
\end{lemma}	
\begin{proof}
Let $x \in X_G$ s.t the orbit of $x$ is finite and $y = F(x)$.

By lemma \ref{lemme:sens1}, $H \cdot F(x) \subseteq F  (G \cdot x )$ therefore the
orbit of $F(x)$ is finite.

\end{proof}	
We now have proven:
\begin{theorem}
Suppose that $H$ is f.p. and acts translation-like on $G$.

Then there is an effective procedure that transforms any SFT $X_H
\subseteq \Sigma^H$ to a SFT $X_G \subseteq (\Sigma \times F)^G$ s.t.:
\begin{itemize}
	\item $X_G$ is empty iff $X_H$ is empty.
	\item If $X_H$ is weakly aperiodic, then $X_G$ is weakly aperiodic.
\end{itemize}	
In particular:
\begin{itemize}
	\item If $H$ has a weakly aperiodic SFT, then so does $G$.
	\item If $H$ has an undecidable domino problem, then so does $G$.
\end{itemize}	
\end{theorem}

In the following we will need the following refinement of lemma
\ref{lemme:sens3}
\begin{lemma}		
	\label{lemme:sens4}	
	Suppose there exists $n$ s.t. if $x \in X_H$ has finite stabilizer
	then $n$ divides $[H:Stab(X_H)]$.

	If $X_G$ contains a periodic point, then there exists a subgroup
	of $G$ of finite index divisible by $n$.
\end{lemma}
\begin{proof}
Suppose that $G$ has  aperiodic point $x$. That is $Stab(x)$ is of
finite index $G$. Then there exists a normal subgroup $K \subseteq H$ of $G$ 
which is of finite index. Write $G = AK$.
By normality, if $g \in K$  and $g'$ in $G$ then $gg' x= g'x$.

Define $g \sim g'$ iff $\exists h, \phi(g, x, h) g'^{-1} \in K$.
It is easy to see that $\sim$ is an equivalence relation on $G$ that
factors into an equivalence relation on $G/K$.

Let $\cal E$ be an equivalence class on $G/K$, and $g$ some element of
$\cal E$.
Let $H_1 = \{ H | \phi(g,x,h) = g\}$.

$H_1$ is a subgroup of $H$ and   by definition $[H:H_1] = |\cal E|$

Now, let $H_2 = \{ h \in H | hF(x) = F(x)\}$.
By definition, $H_2$ contains $H_1$ (remember that $hF(x) = F(\phi(\lambda, x,h) x)$.

Therefore $|{\cal E}| = [H:H_1] = [H:H_2][H_2:H_1]$.

Thus every equivalence class is of cardinality divisible by $n$, and
therefore $G/K$ is of cardinality divisible by $n$.

\end{proof}
\clearpage
\section{Applications}
In this section, we use our main theorem to prove the existence of
aperiodic SFTs on some classes of groups.

\subsection{Amenability}

\begin{theorem}[\cite{Whyte}]
	$\mathbb{F}_2$ acts translation-like on any non amenable f.g. group.
\end{theorem}

\begin{cor}
	Any f.g. non amenable group admits a weakly aperiodic SFT.
\end{cor}
\begin{proof}
Piantadosi	\cite{Piantadosi} exhibited a weakly aperiodic SFT on
$\mathbb{F}_2$, which is a finitely presented group.
\end{proof}
This result can actually be obtained without any reference to
translation-like action, or uniformly finite homology \cite{Block}.
\begin{proof}[Alternative proof]
As $G$ is nonamenable, it admits a (right) paradoxical decomposition, i.e.
$G$ can be partitionned into subsets $A_{-n}, A_{-n-1} \dots A_m\}$ s.t.
\[ G = \biguplus_{i\geq 0}  A_i g_i = \biguplus_{j < 0} A_j g_j\] for some group elements $(g_i)_{-n \leq i\leq m}$
(The use of negative indices while unusual will make the rest of the exposition easier)

Now consider the following subshift

\[
	\begin{array}{rcl}
		X &=& \left\{ x \in \{-n, \dots m\}^G | \forall g  \exists!i \geq
		  0, x_{gg_i^{-1}}  = i \wedge \exists!j > 0, x_{gg_j^{-1}} = j \right\}\\
		 &=& \left\{ x \in \{-n, \dots m\}^G | \forall g  \exists!i \geq
		  0, (g \cdot x)_{g_i^{-1}}  = i \wedge \exists!j > 0, (g \cdot
		  x)_{g_j^{-1}} = j \right\}\\
\end{array}		
\]
$X$ is clearly a SFT.

By defining $x_g = j$ if $g \in A_j$ we see that $X$ is nonempty.

Suppose some $x \in X$ has a stabilizer $H$ of finite index in $G$.
Let $N \subseteq H$ be a normal subgroup of $G$ of finite index.
Then $x$ induces a paradoxical decomposition of the finite group $G/N$, a contradiction.
\end{proof}

The above example of Piantadosi is linked with a homology group
introduced by Chazottes et al.\cite{Gambaudo}: For a SFT in
$\mathbb{Z}^2$ given by Wang tiles\cite{wangpatternrecoII}, one can associate
a finite system of linear equations s.t. if $X$ is nonempty, then the
system has a nontrivial solution with nonnegative coefficients.
This system has as many unknowns as elements of the alphabet, and express
the fact that, for $x \in X$, the frequencies of each element of the
alphabet in $x$ should satisfy some natural conditions (of course
some elements of $X$ do not have frequencies, but frequencies exist
$\mu$-almost surely for $\mu$ an ergodic measure on $X$).
For example consider the example from Piantadosi, taken in  $\mathbb{Z}^2$ with
generators $a$ and $b$, instead of $\mathbb{F}_2$.
\[ X = \{ x \in \{0,1,2\}^{\mathbb{Z}^2} |   x_{ai}  = (x_{i} + 1)\mod 	  3   \wedge x_{i} \not= 1  \iff  x_{bi} = 1\}\]
Then the system of equations we obtain would be
\[
\begin{array}{rcl}
z_0 &=& z_1\\
z_1 &=& z_2\\
z_2 &=& z_0\\
z_1  &=& z_0 + z_2\\
\end{array}
\]
whose only solution is $z_0 = z_1 =  z_2 = 0$, ergo $X$ is empty.  	 

Using suitable analogues of the ergodic theorem for amenable groups,
it is quite likely that the condition of Chazottes et al. could be generalized to any amenable
f.g. group $G$. However the above example shows this is not true
anymore in a nonamenable group.

\subsection{Subgroups of finite index}

If a point of a subshift in $A^G$ is not aperiodic, 
it means that his stabilizer is of finite index in $G$.
If every subshift of finite type is periodic, this implies some
properties on the lattice of subgroups of finite index.
Conversely, we show here that groups for which this lattice behaves 
badly have aperiodic SFTs.

\begin{proposition}
A f.g. and not residually finite group admits a weakly aperiodic SFT.
\end{proposition}
In particular f.g. simple groups admit weakly aperiodic SFTs.
\begin{proof}
As $G$ is not residually finite, there exists a nontrivial element $a \in G$
s.t. every normal subgroup of finite index of $G$ contains $a$.

\[ X = \{ x \in \{0,1,2\}^G | \forall g, (g \cdot x)_\lambda \not= (g \cdot x)_a\}\] 
	X is a SFT and it is easy to see that it is nonempty.
	
	Suppose that $X$ contains a weakly periodic configuration $x$.
	By definition of $X$, $a^{-1} \cdot x \not= x$ therefore the stabilizer of $x$ is a
	proper subgroup of $G$ of finite index which does not contain $a^{-1}$, a contradiction	
\end{proof}	
This gives a different proof on the existence of an aperiodic SFT on
some Baumslag-Solitar groups.

Using translation-like action, we can say a bit more:
\begin{theorem}[\cite{Seward}]
	$\mathbb{Z}$ acts translation-like on any f.g. infinite group.
\end{theorem}

\begin{cor}
	Let $G$ be some infinite f.g. group. Suppose there exists $n$ s.t.
	$G$ does not contain any subgroup of finite index divisible by $n$.
	
	Then $G$ has a weakly aperiodic SFT.
\end{cor}	
\begin{proof}
	$\mathbb{Z}$ acts translation-like on $G$. Now consider 
	\[X_\mathbb{Z} = \{ x \in \{0,1, 2, \dots, n-1\}^\mathbb{Z} | \forall i, 
		  x_{i+1} = x_i + 1 \mod n\}\]
		
$X_\mathbb{Z}$ is a finite subshift where every point is of period exactly $n$.
Now the construction gives a subshift $X_G$. By lemma \ref{lemme:sens4}, $X_G$ has no periodic point.
\end{proof}

In particular $p$-groups are group $G$ where for every element $g\in
G$, $g^{p^k} = \lambda_G$ for some $k$.
As a consequence, $G$ does not contain any group of finite index
divisible by $q$ for any   prime $q > p$:
\begin{cor}
	Infinite f.g. $p$-groups admit aperiodic SFTs.
\end{cor}	

This result is a variation of a construction by Marcinkowski and Nowak
\cite{marcin} of an aperiodic set of tiles on a translation surface on
which the Grigorchuk group acts.
Using a translation-like action, we are able to obtain directly an aperiodic SFT on the Cayley graph of $G$.

\subsection{Direct products and the Domino Problem}

We now give a small example of the relevance of the main theorem to
the domino problem.

\begin{lemma}
	If $H_1$ acts translation-like on $G_1$ and  $H_2$ acts
	translation-like on $G_2$, then $H_1 \times H_2$ acts
	translation-like on $G_1 \times G_2$
\end{lemma}
\begin{cor}
	For any f.g. infinite $G_1, G_2$, $G_1 \times G_2$ has a weakly
	aperiodic SFT and an undecidable domino problem.   
\end{cor}
\begin{proof}
Under the hypothesis, $\mathbb{Z} \times \mathbb{Z}$ acts on $G_1
\times G_2$.
By results of Berger \cite{Berger2}, $\mathbb{Z}^2$ has a weakly
aperiodic SFT and an undecidable domino problem.
\end{proof}	
\begin{cor}
	$\mathbb{Z}^2$ acts translation-like on the Grigorchuk group. 
	Therefore the Grigorchuk group has an aperiodic SFT.
\end{cor}	
\begin{proof}
	$G$ contains a subgroup of finite index of the form $H = H_1 \times H_2$ with
	$H_1, H_2$ infinite~\cite{MuchnikPak}.
	As $H$ is of finite index in $G$, $H$ is finitely generated,
	therefore $H_1$ and $H_2$ are finitely generated as well.	
\end{proof}
Note that the result of Muchnik and Pak \cite{MuchnikPak} was done in
the context of percolation theory: They are interested in the class
$\cal S$ of groups for which there exists $p < 1$ s.t. 
if every vertex of the (undirected) Cayley graph is activated independently with probability $p$, then the
connected component of the identity is infinite a.s. It can be proven that
this property is independent of the generating set so that it is
really a property of the group. The Benjamini-Schramm conjecture states
that $\cal S$ is
exactly the set of all infinite f.g. groups which are not virtually
cyclic, which is the same conclusion as our conjecture.
Note that if the Cayley graph of $G$ contains some copy of some Cayley
graph  of $H$ and $H \in \cal S$, then $G \in \cal S$. In particular
if $H$ acts translation-like on $G$ and $H \in \cal S$ then $G \in \cal S$.
We do not know if there is some link between the existence of a
weakly aperiodic SFT and percolation.

\clearpage
\paragraph{}
We will now obtain a better result, that corresponds to our version of
the theorem of Ballier and Stein \cite{BallierStein}

\begin{lemma}
Let $\psi: G \mapsto G'$ a onto morphism s.t. $\mathbb{Z}$ acts
translation-like on $G'$.
Then there exists a transversal $K$ (i.e. $\psi$ is injective on
$K$ and $\psi(K) = G'$) s.t $\mathbb{Z}$ acts translation-like on $K$.
\end{lemma}	
(Note that the definition of a translation-like action can be defined
even if $K$ is not a group)
\begin{proof}
First, $\mathbb{Z}$ acts translation-like on $G'$. Let $S_{G'}$ be 
the finite set that witnesses it when the generators of $\mathbb{Z}$
are chosen to be $-1$ and $+1$. Write $S_{G'} = \{ a_1 \dots a_n\}$.

Let $R \subseteq G'$ be a set of representations of the free action,
that is every element of $G'$ can be uniquely written $g = r \circ h$
for some $h \in \mathbb{Z}$ and $r \in R$.

Now we fix some notation.
A \emph{biinfinite} sequence in a group $H$ is  a map 
$w : \mathbb{Z}^\times \rightarrow H$ where $\mathbb{Z}^\times$ is the
set of nonzero integers.
For such a sequence, we write
$w(0) = \lambda$,  $w(n) = w_1 \dots w_n$ for $n \geq 0$ and
$w(n) = w_{-1} w_{-2} \dots w_n$ for $n < 0$, 

Now let $r\in R$.
As $r \circ n = r \circ 1 \circ 1 \circ \dots \circ 1$ for $n \geq 0$
(and similarly for $n < 0$), there exists an infinite sequence
$w^{r}$ with values in $S_{G'}$ s.t.
$r \circ n = r w^{r}(n)$.
Basically, an element of $r \circ \mathbb{Z}$  is represented by the label
on the path from $r$ to this element following the copy of
the Cayley graph of $\mathbb{Z}$. Note that there is a unique path
from $r$ to this element, as the Cayley graph of $\mathbb{Z}$ is a tree.

Now we take representatives of all elements of $S_{G'}$ and all elements
of $R$: Let $\theta: R \cup S_{G'} \rightarrow G$ s.t. $\psi\theta(g)
= g$ whenever it is defined.

Now we take $K = \{ \theta(r) \theta(w^r)(n) | r \in R, n \in
  \mathbb{Z}\}$, where $\theta(w^r)_i = \theta(w^r_i)$

It is easy to see that  $K$ is a transversal and that all elements in
the above formula are distinct.

Furthermore $\mathbb{Z}$ acts on $K$ by 
$\theta(r) \theta(w^r(n)) \circ m = \theta(r) \theta(w^r(n+m))$.
This action is clearly free.

Furthermore, for all $k \in K$, $k \circ 1 = kb$ for some $b \in
\theta(S_{G'}) \cup \theta(S_{G'})^{-1}$, and the same is true for $-1$, so
that $\theta(S_{G'}) \cup \theta(S_{G'})^{-1}$ witnesses the fact that
$\mathbb{Z}$ acts translation-like on $K$.
\end{proof}	
The same could be done replacing $\mathbb{Z}$ by any group which
admits a Cayley graph which is a tree, i.e. a free group.
The result is false in general if $\mathbb{Z}$ is replaced by a group
which is not free, take for example the morphism from the free group
to $\mathbb{Z}^2$.

\begin{cor}
If $H \subseteq Z(G)$ is finitely generated and $G/H$ is f.g. infinite, 
then $\mathbb{Z} \times H$ acts translation-like on $G$.
\end{cor}
\begin{proof}
Use the previous lemma and define $k a \circ (n,b) = (k \circ n) (ab)$
for $k \in K, a,b \in H$ and $n \in \mathbb{Z}$.
It is easy to see that it is indeed an action, and it is translation-like
\end{proof}

\begin{cor}[\cite{BallierStein}]
	If $Z(G)$ is f.g.,  $\mathbb{Z} \subseteq Z(G)$ and $G/Z(G)$ is
	f.g. infinite, then $G$ has a weakly aperiodic SFT and an undecidable domino problem
\end{cor}

\begin{cor}
	If a f.g. group $G$ admits $\mathbb{Z}$ as a normal subgroup, then it
	has undecidable domino problem unless it is virtually cyclic.
	
	More precisely, some Baumslag Solitar group
	$B(1,n)$, $n \not= 0$, acts translation-like on $G$.   
\end{cor}
\begin{proof}
Let $H \subseteq G$ with $H \simeq \mathbb{Z}$ and let $a$ be
the generator of $H$.
As $H$ is normal, for any element $t$ of $G$, we have $tat^{-1} = a^n$
for some $n$ that depends on $t$.

If $tat^{-1} = a^n$ for some $n \not\in \{1,-1\}$, then the subgroup
generated by $t$ and $a$ is isomorphic to the Baumslag Solitar group
$B(1,n)$ and therefore $G$ admits a subgroup with undecidable domino
problem and therefore has undecidable domino problem.

Otherwise $tat^{-1} = a^{\pm 1}$ for all $t \in G$.
Then $G' = \{ t | tat^{-1} = a\}$ is a subgroup of $G$ that contains
$H$ of index at most $2$ in $G$.
$H$ is in the center of $G'$, so that $\mathbb{Z}  \times H \simeq
\mathbb{Z}^2$ acts translation-like on $G'$, therefore it acts
translation like on $G$.
\end{proof}

\section{Further generalizations}

First, let remark that the whole construction works as well starting
from a finitely presented monoid rather than a finitely presented group.
In fact, the result of Ballier and Stein \cite{BallierStein} build a $\mathbb{Z}\times
\mathbb{N}$-action on any f.g nilpotent group which is not virtually
$\mathbb{Z}$.
We do not go into details on this generalization as some details are
quite cumbersome. In particular a translation-like action by
$\mathbb{N}$ does not partition the Cayley graph into copies of
$\mathbb{N}$ but into some kind of trees.

\paragraph{}
Note that all our proofs exploit the aperiodicity of $X_H$ to prove
that $X_G$ is aperiodic. There is another way to do this, by forcing
the copies of $H$ inside $G$ to be infinite.

For example, for a finite subset ${\cal T}$ of $S^\star \times S^\star$, let's consider the
following variant of $X$

\[
	X  = \left\{ x \in F^G | \forall g \in G,
	  \left.\begin{array}{l}
			\forall (h,h') \in {\cal R},
	  \Phi(g, x, h) = \Phi(g,x,h') \\
	  \forall (h,h') \in {\cal T},
	  \Phi(g,x,h) \not= \Phi(g,x,h')
  \end{array}
\right.\right \}
     \]

\emph{Isolated} groups \cite{Cornulier} (also called groups with
finite absolute presentation) admit a presentation in this term: If
$G$ is an isolated group with generators $S$, there exists a finite set of relations $\cal R$
and antirelations $\cal T$  s.t. $G$ is the only f.g. group generated by
$S$ that satisfy these relations and antirelations.
As a consequence, any copy of $H$ (that is $\{\Phi(g,x,h), h \in H\}$
for some $g$) should be in bijection with  $H$, and therefore infinite.

We do not give more details as all these (infinite) groups are not
residually finite, therefore this result is subsumed by the
construction of an aperiodic SFT on any nonresidually finite group.

\paragraph{}
An other interesting direction is the undecidability of the \emph{periodic
domino problem}: decide, given a SFT $X$ over a group $G$, if $X$ is
weakly aperiodic. Most of our work here fail for this problem: it is
not true that  $X_G$ is weakly aperiodic iff $X_H$ is.

\section{Open Problems}

Carroll and Penland \cite{Carroll} showed that the existence of an aperiodic
SFT is also closed under another notion of containment:
If $H$ is a finitely generated normal subgroup of $G$ and $G/H$ has a
weakly aperiodic SFT (resp. has an undecidable domino problem), then $G$ does.
Finite generation of $H$ is essential here: $\mathbb{F}_2$ contains
$\mathbb{Z}^2$ as a quotient, but $\mathbb{Z}^2$ has an undecidable
word problem and $\mathbb{F}_2$ does not.

Can this be generalized as well by asking for less than a quotient ?

To tackle the two conjectures stated above, we introduce now a new
conjecture:

\begin{conj}
If $G$ is a nontrivial f.g. group, then some nontrivial one-relator
group acts translation-like on $G$.
\end{conj}
Here nontrivial means not virtually cyclic.
Note that the amenable nontrivial one-relator groups are the Baumslag
Solitar groups, so that this conjecture may be restated as: If $G$ is
a nontrivial amenable group, then some Baumslag-Solitar group acts
translation-like on $G$.

This conjecture would imply that every nontrivial group admits a
weakly aperiodic SFT, and the Benjamini-Schramm conjecture stated above.

We give the lamplighter group as a potential counterexample to the 
three conjectures of this paper.

\end{document}